\documentclass[copyright,creativecommons]{eptcs}
\providecommand{\event}{GandALF 2010}
\usepackage[cp1250] {inputenc}
\usepackage{amsmath}
\usepackage{amssymb}
\usepackage{graphicx}
\usepackage{amsthm}
\begin{document}
\newtheorem{definition}{Definition}
\newtheorem{notation}{Notation}
\newtheorem{corollary}{Corollary}
\newtheorem{lemma}{Lemma}
\newtheorem{theorem}{Theorem}
\title{Unitary Noise and the Mermin-GHZ Game}
\author{Ivan Fial\'ik \thanks{This work has been supported by the research project MSM0021622419.}
\institute{Faculty of Informatics\\
Masaryk University\\
Brno, Czech Republic}
\email{60488@mail.muni.cz}
}
\def\titlerunning{Unitary Noise and the Mermin-GHZ Game}
\def\authorrunning{Ivan Fial\'ik}
\maketitle

\begin{abstract}
Communication complexity is an area of classical computer science which studies how much communication is necessary to solve various distributed computational problems. Quantum information processing can be used to reduce the amount of communication required to carry out some distributed problems. We speak of pseudo-telepathy when it is able to completely eliminate the need for communication. Since it is generally very hard to perfectly implement a quantum winning strategy for a pseudo-telepathy game, quantum players are almost certain to make errors even though they use a winning strategy. After introducing a model for pseudo-telepathy games, we investigate the impact of erroneously performed unitary transformations on the quantum winning strategy for the Mermin-GHZ game. The question of how strong the unitary noise can be so that quantum players would still be better than classical ones is also dealt with.
\end{abstract}
\section{Introduction}
Since quantum entanglement provides us with strong non-local correlations that cannot be achieved in the classical world, one can ask whether it can be used even to solve some distributed problems without any form of direct communication between the parties. Of course, we are interested only in such problems for which this does not hold in the classical world. On one hand, the answer is negative if we consider the standard communication complexity model \cite{scqrtdc} in which the parties compute a value of some function on their inputs and the whole result of the computation must become known to at least one party. Otherwise, faster-than-light communication would be possible which would contradict the Relativity Theory. On the other hand, if each party has its own input, computes its own output and we are interested only in non-local correlations between the inputs and the outputs, then the answer is positive. Such problems are often described using a terminology of the game theory and they are usually called pseudo-telepathy games.

In order to be able to describe what a pseudo-telepathy game \cite{qp} is, we explain at first what we mean by the term two party game. A \emph{two party game} $G$ is a sextuple $$G = (X, Y, A, B, P, W)$$ where $X, Y$ are \emph{input sets}, $A, B$ are \emph{output sets}, $P$ is a subset of $X \times Y$ known as a \emph{promise} and $W \subseteq X \times Y \times A \times B$ is a relation among the input sets and the output sets which is called a \emph{winning condition}. Before the game begins, the parties, usually called Alice and Bob, are allowed to discuss strategy and exchange any amount of classical information, including values of random variables. They may also share an unlimited amount of quantum entanglement. Afterwards, Alice and Bob are separated from each other and they are not allowed to communicate any more till the end of the game. In one \emph{round of the game}, Alice is given an input $x \in X$ and she is required to produce an output $a \in A$. Similarly, Bob is given an input $y \in Y$ and he is required to produce an output $b \in B$. The pairs $(x, y)$ and $(a, b)$ are called a \emph{question} and an \emph{answer}, respectively. Alice and Bob \emph{win the round} if either $(x, y) \notin P$ or $(x, y, a, b) \in W$. Alice and Bob \emph{win the game} if they have won all the rounds of it. A \emph{strategy} of Alice and Bob is \emph{winning} if it always allows them to win. 

We say that a two-party game is \emph{pseudo-telepathic} if there is no classical winning strategy, but there is a winning strategy, provided Alice and Bob share quantum entanglement. The origin of this term can be explained in the following way. Suppose that scientists who know nothing about quantum computing witness Alice and Bob playing some pseudo-telepathy game. More precisely, suppose that the players are very far from each other, are given their inputs at the same time and have to produce their outputs in time shorter than time required by light to travel between them. If Alice and Bob answer correctly in a sufficiently long sequence of rounds, the scientists will conclude that Alice and Bob can communicate somehow. One of possible explanations will be that the players are endowed with telepathic powers. A survey of pseudo-telepathy games can be found in \cite{qp}. The definition of these games can be easily generalized to more than two players.

A classical strategy $s$ for a pseudo-telepathy game $G$ is \emph{deterministic} if there are functions $s_{A} : X \to A$ and $s_{B}: Y \to B$ such that for each question $(x, y) \in X \times Y$, the only possible answer of Alice and Bob is the pair $(s_{A}(x), s_{B}(y))$. The \emph{success of a deterministic strategy} is defined as the proportion of questions from the promise $P$ for which it produces a correct answer. Clearly, this number can by interpreted as the probability that the strategy succeeds on a given question which is chosen uniformly and randomly. We denote with $\omega_{d}(G)$ the maximal success of a deterministic strategy for the game $G$: $$\omega_{d}(G) = \max_{s} \frac{\{(x, y) \in P \ | \ (x, y, s_{A}(x), s_{B}(y)) \in W\}}{|P|}.$$

Alice and Bob can also use a randomized strategy for $G$. Any randomized strategy can be seen as a probability distribution over a finite set of deterministic strategies. Therefore, if questions are chosen uniformly and randomly, the probability of winning a round of the game $G$ using a randomized strategy cannot be greater than $\omega_{d}(G)$ \cite{qp}.

We will consider in this paper quantum strategies of the following simple form. Alice and Bob share some entangled state $|\varphi\rangle$ of a quantum system consisting of two quantum bits. After they have been given their inputs $x \in X$ and $y \in Y$, the players apply on their parts of $|\varphi\rangle$ unitary transformations $U_{x}$ and $U_{y}$, respectively. These transformations can always be decomposed into rotations around one of the axes in the Bloch sphere and CNOT transformations \cite{qcqi}. Then the players perform measurements $M_{x}$ and $M_{y}$ on their parts of $|\varphi\rangle$ which give them their outputs $a \in A$ and $b \in B$, respectively. It was shown in \cite{mesdrfpt} that no strategy of this type for two-party games can be winning because the dimension of $|\varphi\rangle$ is too small. However, an analogue of this result does not hold for games for more players. This can be easily shown by examining properties of the quantum winning strategy for the Mermin-GHZ game.

This paper investigates the impact of unitary noise on the success of the quantum winning strategy for the Mermin-GHZ game. A general definition of this problem is given in the next section. The Mermin-GHZ game and the quantum winning strategy for it are described in Section 3. The remaining sections are devoted to the study of the influence of erroneously performed unitary transformations on the success of this strategy.
\section{Pseudo-Telepathy in the Presence of Unitary Noise}
An experimental implementation of a quantum winning strategy is generally very hard to be perfect. The players may perform imperfectly the unitary transformations required by the winning strategy or they may not be able to keep the required entangled quantum state. Moreover, their measurement devices may fail to give a correct outcome or may fail to give an outcome at all. This paper is focused on the imperfections of the first type and deals with the question how they affect the quantum winning strategy for the Mermin-GHZ game. It continues in the research initiated in \cite{nmghzg} where the impact of the imperfections of the second type has been analyzed.

Remind that the unitary transformations $U_{x}$ and $U_{y}$ can be implemented using only rotations around one of the axes in the Bloch sphere and CNOT transformations. The players are supposed not to be able to perform the required unitary transformations exactly. More precisely, we will focus on what happens if one of them, say Alice, performs some of her rotations around the right axes, but with an incorrect rotation angles. Thus, her erroneous unitary transformation is specified by an $n$-tuple $$\alpha_{x}^{\epsilon} = (\alpha_{x_{1}}^{\epsilon}, \ldots, \alpha_{x_{n}}^{\epsilon})$$ of rotation angles such that $$\alpha_{x_{i}}^{\epsilon} = \alpha_{x_{i}} + \epsilon_{i}$$ where $i \in \{1, \ldots, n\}$ and $\alpha_{x} = (\alpha_{x_{1}}, \ldots, \alpha_{x_{n}})$ is an $n$-tuple of Alice's correct rotation angles. The $n$-tuple $\epsilon = (\epsilon_{1}, \ldots, \epsilon_{n})$ of rotation angles is called an error. In the following we will focus on two basic types of errors, systematic errors and random errors.

An error is said to be \emph{systematic} if it is constant for a given unitary transformation. Let $p_{\epsilon}^{A}(x, y)$ ($p_{\epsilon}^{B}(x, y)$) be a probability of the event that Alice and Bob obtain after performing the quantum winning strategy for $G$ with systematic error $\epsilon$, which affects Alice's (Bob's) unitary transformation, a state corresponding to a correct answer to $(x, y)$.

An error is said to be random with bound $\delta = (\delta_{1}, \ldots, \delta_{n})$ if for each $i \in \{1, \ldots, n\}$ the error $\epsilon_{i}$ is chosen uniformly and randomly from the interval $[-\delta_{i}, \delta_{i}]$. Let $p_{\delta}^{A}(x, y)$ ($p_{\delta}^{B}(x, y)$) be a probability of the event that Alice and Bob obtain after performing the quantum winning strategy for $G$ with random error with bound $\delta$, which affects Alice's (Bob's) unitary transformation, a state corresponding to a correct answer to $(x, y)$.
\section{Mermin-GHZ Game}
Alice, Bob and Charles have each one bit as an input with the promise that the parity of the input bits is 0 \cite{btwi}. We denote the input bits $x_{1}$, $x_{2}$ and $x_{3}$. The task for each player is to produce one bit so that the parity of the output bits is equal to the disjunction of the input bits. Thus, if $a_{1}$, $a_{2}$ and $a_{3}$ are the outputs, then the equation $a_{1} \oplus a_{2} \oplus a_{3} = x_{1} \lor x_{2} \lor x_{3}$ must hold.
 
In a quantum winning strategy for the Mermin-GHZ game Alice, Bob and Charles share the entangled state $|\varphi\rangle = \frac{1}{\sqrt{2}}(|000\rangle + |111\rangle)$ \footnote{A measurement of one of the qubits in the computational basis yields with probability $\frac{1}{2}$ the result 0 and with probability $\frac{1}{2}$ the result 1. The measurement causes the remaining two qubits to collapse into the state which has been measured. These peculiar correlations, which are a characteristic property of entangled quantum states, appear regardless of the distance between the individual qubits and have no counterpart in the classical world.}. After the players have received their inputs $x_{1}, x_{2}, x_{3}$, respectively, each of them does the following:
\begin{enumerate}
   \item Applies to his or her register the unitary transformation $U$ which is defined as $$U(|0\rangle) = |0\rangle,$$$$U(|1\rangle) = e^{\frac{\pi \imath x_{i}}{2}}|1\rangle.$$
   \item Applies to his or her register the Hadamard transformation $H$.
   \item Performs a measurement in the computational basis on his or her register. Outputs the bit $a_{i}$ which he or she has measured.
\end{enumerate}

Observe that if $x_{i} = 1$, the unitary transformation $U$ can be seen as the rotation of $\frac{\pi}{2}$ around the $z$ axis in the Bloch sphere, otherwise it is the identity transformation. Similarly, the Hadamard transformation can be seen as the rotation of $\frac{\pi}{2}$ around the $y$ axis. Observe further that the Mermin-GHZ game and also the quantum winning strategy are completely symmetric. Therefore, it is sufficient to compute success probabilities for errors made by one fixed player, say Alice. Success probabilities for errors made by the other players have to be the same.

It is not very hard to see that the best possible classical strategy for the Mermin-GHZ game enables the players to win with probability $\frac{3}{4}$.
\section{Errors in the First Transformation}
\subsection{Systematic errors}
\begin{theorem}
If Alice performs the unitary transformation $U$ imperfectly with systematic error $\epsilon$, then the quantum winning strategy for the Mermin-GHZ game gives a correct answer with probability
\begin{eqnarray*}
               P_{\epsilon}(x_{1}, x_{2}, x_{3}) = 
                  \left\{ \begin{array}{ll}
                      1 \ & \ \text{if $x_{1} = 0$}\\
                      \frac{1 + \cos\epsilon}{2} \ & \text{otherwise}\\
                  \end{array} \right.
\end{eqnarray*}
for any inputs $x_{1}$, $x_{2}$ and $x_{3}$ satisfying the promise. It gives a correct answer with probability $$P_{\epsilon} = \frac{1}{2} + \frac{1 + \cos\epsilon}{4}$$ for a question chosen uniformly and randomly from P.
\end{theorem}
\begin{proof}
Let Alice, Bob and Charles be given bits $x_{1}, x_{2}, x_{3}$, respectively, such that $\sum_{i = 1}^{3}x_{i}$ is divisible by $2$. Suppose that while Bob and Charles perform all the steps required by the
quantum winning strategy for the Mermin-GHZ game correctly, Alice deviates in the first step from the winning strategy by performing the unitary transformation
$$U_{\epsilon}(|0\rangle) = |0\rangle,$$$$U_{\epsilon}(|1\rangle) = e^{(\frac{\pi}{2} + \epsilon) \imath x_{i}}|1\rangle$$ where $\epsilon$ is a constant. Let $t_{0}, \ldots, t_{7}$ be the diagonal elements of the resulting density matrix $\rho_{E}$. Since the players perform their measurements in the computational basis, the probability $P_{\epsilon}(x_{1}, x_{2}, x_{3})$ of the event that they obtain outcomes corresponding to a correct answer to the question $(x_{1}, x_{2}, x_{3})$ can be simply computed as a sum of some of the numbers $t_{0}, \ldots, t_{7}$.

If $x_{1} = 0$, the players actually perform the quantum winning strategy for the Mermin-GHZ game. Therefore, they win with probability 1.

Otherwise, the players win if and only if they measure one of the states $$|001\rangle, |010\rangle, |100\rangle, |111\rangle.$$ A straightforward but slightly tedious computation yields $$t_{1} = t_{2} = t_{4} = t_{7} = \frac{1}{8}(1 + \cos\epsilon).$$ 

It follows that
\begin{eqnarray*}
               P_{\epsilon}(x_{1}, x_{2}, x_{3}) = 
                  \left\{ \begin{array}{ll}
                      1 \ & \ \text{if $x_{1} = 0$}\\
                      \frac{1 + \cos\epsilon}{2} \ & \text{otherwise}\\
                  \end{array} \right.
\end{eqnarray*}
and $$P_{\epsilon} = \frac{1}{2} + \frac{1 + \cos\epsilon}{4}.$$
\end{proof}
Quantum players are better than classical ones in the presence of a systematic error in the first unitary transformation if $\epsilon \in (-\frac{\pi}{2}, \frac{\pi}{2})$.
\begin{figure}[h]
    \centering
    \includegraphics[width = 7cm]{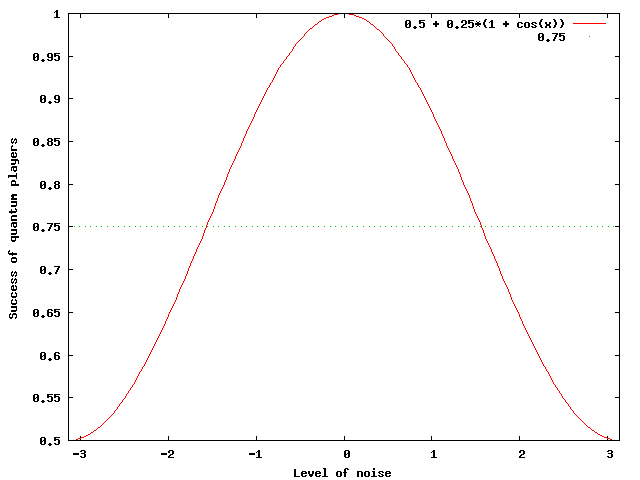}
    \includegraphics[width = 7cm]{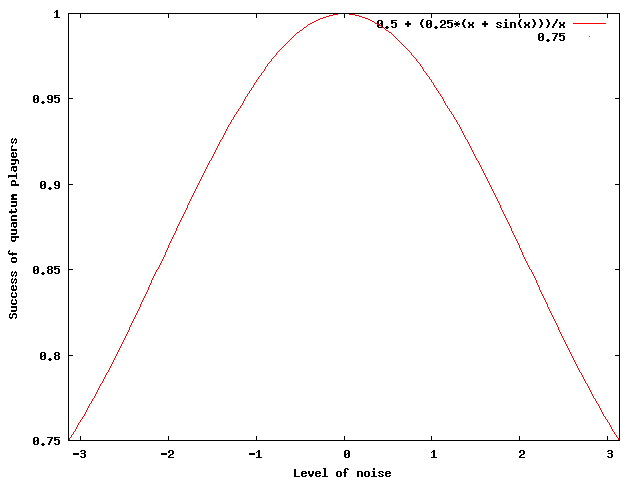}
\caption{Dependence of the success of quantum players on the level of unitary noise in the presence of a systematic (the graph on the left) and a random error made by one player in the first unitary transformation.}   
\end{figure}
\subsection{Random errors}
\begin{theorem}
If Alice performs the unitary transformation $U$ imperfectly with random error with bound $\delta$, then the quantum winning strategy for the Mermin-GHZ game gives a correct answer with probability
\begin{eqnarray*}
               P_{\delta}(x_{1}, x_{2}, x_{3}) = 
                  \left\{ \begin{array}{ll}
                      1 \ & \ \text{if $x_{1} = 0$}\\
                      \frac{\delta + \sin\delta}{2\delta} \ & \text{otherwise}\\
                  \end{array} \right.
\end{eqnarray*}
for any inputs $x_{1}$, $x_{2}$ and $x_{3}$ satisfying the promise. It gives a correct answer with probability $$P_{\delta} = \frac{1}{2} + \frac{\delta + \sin\delta}{4\delta}$$ for a question chosen uniformly and randomly from P.
\end{theorem}
\begin{proof}
Let Alice, Bob and Charles be given bits $x_{1}, x_{2}, x_{3}$, respectively, such that $\sum_{i = 1}^{3}x_{i}$ is divisible by $2$. Suppose that while Bob and Charles perform all the steps required by the quantum winning strategy for the Mermin-GHZ game correctly, Alice deviates in the first step from the winning strategy by performing the unitary transformation $$U_{\epsilon}(|0\rangle) = |0\rangle,$$$$U_{\epsilon}(|1\rangle) = e^{(\frac{\pi}{2} + \epsilon) \imath x_{i}}|1\rangle$$ where $\epsilon$ is chosen uniformly and randomly from the interval $[-\delta, \delta]$. The probability $P_{\delta}(x_{1}, x_{2}, x_{3})$ of the event that they obtain outcomes corresponding to a correct answer to the question $(x_{1}, x_{2}, x_{3})$ can be computed as $$P_{\delta}(x_{1}, x_{2}, x_{3}) = \frac{1}{2\delta}\int\limits_{-\delta}^{\delta}P_{\epsilon}(x_{1}, x_{2}, x_{3})\, \mathrm{d} \epsilon$$ where $P_{\epsilon}(x_{1}, x_{2}, x_{3})$ is the probability of a correct answer in the presence of a systematic error in the first transformation. 

It follows that
\begin{eqnarray*}
               P_{\delta}(x_{1}, x_{2}, x_{3}) = 
                  \left\{ \begin{array}{ll}
                      1 \ & \ \text{if $x_{1} = 0$}\\
                      \frac{\delta + \sin\delta}{2\delta} \ & \text{otherwise}\\
                  \end{array} \right.
         \end{eqnarray*}
and $$P_{\alpha} = \frac{1}{2} + \frac{\delta + \sin\delta}{4\delta}.$$
\end{proof}

Quantum players are better than classical ones in the presence of a random error in the first unitary transformation if $\delta \in (-\pi, \pi)$.
\section{Errors in the Second Transformation}
\subsection{Systematic errors}
\begin{theorem}
If Alice performs the Hadamard transformation $H$ imperfectly with systematic error $\epsilon$, then the quantum winning strategy for the Mermin-GHZ game gives a correct answer with probability $$P_{\epsilon}(x_{1}, x_{2}, x_{3}) = \frac{1}{2} + \sin(\frac{\pi}{4} + \frac{\epsilon}{2})\cos(\frac{\pi}{4} + \frac{\epsilon}{2})$$
for any inputs $x_{1}$, $x_{2}$ and $x_{3}$ satisfying the promise. It gives a correct answer with probability $$P_{\epsilon} = \frac{1}{2} + \sin(\frac{\pi}{4} + \frac{\epsilon}{2})\cos(\frac{\pi}{4} + \frac{\epsilon}{2})$$ for a question chosen uniformly and randomly from P.
\end{theorem}
\begin{proof}
Let Alice, Bob and Charles be given bits $x_{1}, x_{2}, x_{3}$, respectively, such that $\sum_{i = 1}^{3}x_{i}$ is divisible by $2$. Suppose that while Bob and Charles perform all the steps required by the quantum winning strategy for the Mermin-GHZ game correctly, Alice deviates in the second step from the winning strategy by performing the unitary transformation 
\begin{displaymath}
H_{\epsilon} = \left( \begin{array}{cc}
    \cos(\frac{\pi}{4} + \frac{\epsilon}{2}) & \sin(\frac{\pi}{4} + \frac{\epsilon}{2}) \\
    \sin(\frac{\pi}{4} + \frac{\epsilon}{2}) & -\cos(\frac{\pi}{4} + \frac{\epsilon}{2}) \\
  \end{array} \right) 
\end{displaymath}  
where $\epsilon$ is a constant. Let $t_{0}, \ldots, t_{7}$ be the diagonal elements of the resulting density matrix $\rho_{E}$. Since the players perform their measurements in the computational basis, the probability $P_{\delta}(x_{1}, x_{2}, x_{3})$ of the event that they obtain outcomes corresponding to a correct answer to the question $(x_{1}, x_{2}, x_{3})$ can be simply computed as a sum of some of the numbers $t_{0}, \ldots, t_{7}$.

If $x_{1} = x_{2} = x_{3} = 0$, the players win if and only if they measure one of the states $$|000\rangle, |011\rangle, |101\rangle, |110\rangle.$$ A straightforward but slightly tedious computation yields $$t_{0} = t_{3} = t_{5} = t_{6} = \frac{1}{4}(\frac{1}{2} + \sin(\frac{\pi}{4} + \frac{\epsilon}{2})\cos(\frac{\pi}{4} + \frac{\epsilon}{2})).$$

If some $x_{i} = 1$, where $i \in \{1, 2, 3\}$, the players win if and only if they measure one of the states $$|001\rangle, |010\rangle, |100\rangle, |111\rangle.$$ A straightforward but slightly tedious computation yields $$t_{1} = t_{2} = t_{4} = t_{7} = \frac{1}{4}(\frac{1}{2} + \sin(\frac{\pi}{4} + \frac{\epsilon}{2})\cos(\frac{\pi}{4} + \frac{\epsilon}{2})).$$ 

It follows that $$P_{\epsilon}(x_{1}, x_{2}, x_{3}) = P_{\epsilon} = \frac{1}{2} + \sin(\frac{\pi}{4} + \frac{\epsilon}{2})\cos(\frac{\pi}{4} + \frac{\epsilon}{2}).$$
\end{proof}

Quantum players are better than classical ones in the presence of a systematic error in the Hadamard transformation if $\epsilon \in (-\frac{\pi}{3}, \frac{\pi}{3})$.
\begin{figure}[h]
    \centering
    \includegraphics[width = 7cm]{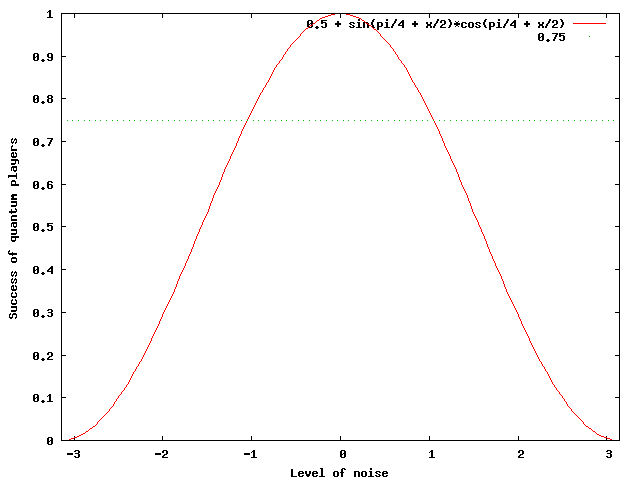}
    \includegraphics[width = 7cm]{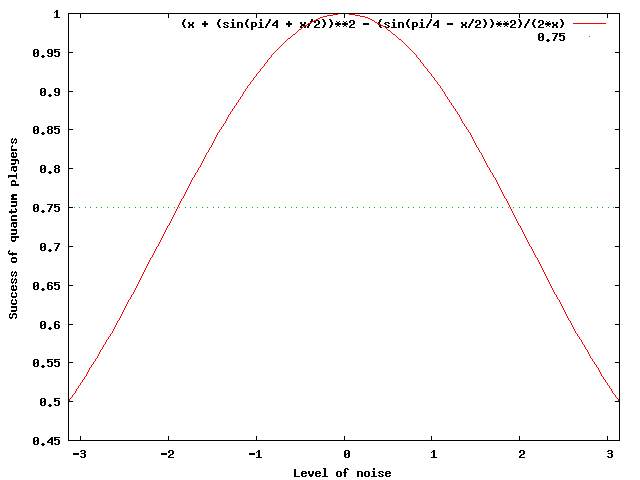}
\caption{Dependence of the success of quantum players on the level of unitary noise in the presence of a systematic (the graph on the left) and a random error made by one player in the second unitary transformation.}     
\end{figure}
\subsection{Random errors}
\begin{theorem}
If Alice performs the Hadamard transformation $H$ imperfectly with random error with bound $\delta$, then the quantum winning strategy for the Mermin-GHZ game gives a correct answer with probability $$P_{\delta}(x_{1}, x_{2}, x_{3}) = \frac{1}{2\delta}(\delta + \sin^{2}(\frac{\pi}{4} + \frac{\delta}{2}) - \sin^{2}(\frac{\pi}{4} - \frac{\delta}{2}))$$
for any inputs $x_{1}$, $x_{2}$ and $x_{3}$ satisfying the promise. It gives a correct answer with probability $$P_{\delta} = \frac{1}{2\delta}(\delta + \sin^{2}(\frac{\pi}{4} + \frac{\delta}{2}) - \sin^{2}(\frac{\pi}{4} - \frac{\delta}{2}))$$ for a question chosen uniformly and randomly from P.
\end{theorem}
\begin{proof}
Let Alice, Bob and Charles be given bits $x_{1}, x_{2}, x_{3}$, respectively, such that $\sum_{i = 1}^{3}x_{i}$ is divisible by $2$. Suppose that while Bob and Charles perform all the steps required by the quantum winning strategy for the Mermin-GHZ game correctly, Alice deviates in the second step from the winning strategy by performing the unitary transformation 
\begin{displaymath}
H_{\epsilon} = \left( \begin{array}{cc}
    \cos(\frac{\pi}{4} + \frac{\epsilon}{2}) & \sin(\frac{\pi}{4} + \frac{\epsilon}{2}) \\
    \sin(\frac{\pi}{4} + \frac{\epsilon}{2}) & -\cos(\frac{\pi}{4} + \frac{\epsilon}{2}) \\
  \end{array} \right) 
\end{displaymath}  
where $\epsilon$ is chosen uniformly and randomly from the interval $[-\delta, \delta]$. The probability $P_{\delta}(x_{1}, x_{2}, x_{3})$ of the event that they obtain outcomes corresponding to a correct answer to the question $(x_{1}, x_{2}, x_{3})$ can be computed as $$P_{\delta}(x_{1}, x_{2}, x_{3}) = \frac{1}{2\delta}\int\limits_{-\delta}^{\delta}P_{\epsilon}(x_{1}, x_{2}, x_{3})\, \mathrm{d} \epsilon$$ where $P_{\epsilon}(x_{1}, x_{2}, x_{3})$ is the probability of a correct answer in the presence of a systematic error in the Hadamard transformation. 

It follows that $$P_{\epsilon}(x_{1}, x_{2}, x_{3}) = P_{\epsilon} = \frac{1}{2\delta}(\delta + \sin^{2}(\frac{\pi}{4} + \frac{\delta}{2}) - \sin^{2}(\frac{\pi}{4} - \frac{\delta}{2})).$$
\end{proof}

Quantum players are better than classical ones in the presence of a random error in the Hadamard transformation if $\delta \in (-1.896, 1.896)$. This is only an approximate result since it has been computed using numerical methods.
\section{Errors in Both Transformations}
\subsection{Systematic errors}
\begin{theorem}
If Alice performs both transformations $U$, $H$ imperfectly with systematic errors $\epsilon_{1}$, $\epsilon_{2}$, respectively, then the quantum winning strategy for the Mermin-GHZ game gives a correct answer with probability \begin{eqnarray*}
               P_{\epsilon}(x_{1}, x_{2}, x_{3}) = 
                  \left\{ \begin{array}{ll}
                      \frac{1}{2} + \sin(\frac{\pi}{4} + \frac{\epsilon_{2}}{2})\cos(\frac{\pi}{4} + \frac{\epsilon_{2}}{2}) \ & \ \text{if $x_{1} = 0$}\\
                      \frac{1}{2} + \sin(\frac{\pi}{4} + \frac{\epsilon_{2}}{2})\cos(\frac{\pi}{4} + \frac{\epsilon_{2}}{2})\cos\epsilon_{1} \ & \text{otherwise.}\\
                  \end{array} \right.
\end{eqnarray*}
for any inputs $x_{1}$, $x_{2}$ and $x_{3}$ satisfying the promise. It gives a correct answer with probability $$P_{\epsilon} = \frac{1}{2}(1 + \sin(\frac{\pi}{4} + \frac{\epsilon_{2}}{2})\cos(\frac{\pi}{4} + \frac{\epsilon_{2}}{2})(1 + \cos\epsilon_{1}))$$ for a question chosen uniformly and randomly from P.
\end{theorem}
\begin{proof}
Let Alice, Bob and Charles be given bits $x_{1}, x_{2}, x_{3}$, respectively, such that $\sum_{i = 1}^{3}x_{i}$ is divisible by $2$. Suppose that while Bob and Charles perform all the steps required by the quantum winning strategy for the Mermin-GHZ game correctly, Alice deviates in the first and the second step from the winning strategy by performing the unitary transformations 
\begin{eqnarray*}
\left. \begin{matrix}
U_{\epsilon_{1}}(|0\rangle) = |0\rangle, \\
U_{\epsilon_{1}}(|1\rangle) = e^{(\frac{\pi}{2} + \epsilon_{1}) \imath x_{i}}|1\rangle \\
\end{matrix} \right.,~~
H_{\epsilon_{2}} = \left( \begin{matrix}
    \cos(\frac{\pi}{4} + \frac{\epsilon_{2}}{2}) & \sin(\frac{\pi}{4} + \frac{\epsilon_{2}}{2}) \\
    \sin(\frac{\pi}{4} + \frac{\epsilon_{2}}{2}) & -\cos(\frac{\pi}{4} + \frac{\epsilon_{2}}{2}) \\
  \end{matrix} \right), 
\end{eqnarray*}  
respectively, where $\epsilon_{1}$, $\epsilon_{2}$ are constants. Let $t_{0}, \ldots, t_{7}$ be the diagonal elements of the resulting density matrix $\rho_{E}$. Since the players perform their measurements in the computational basis, the probability $P_{\delta}(x_{1}, x_{2}, x_{3})$ of the event that they obtain outcomes corresponding to a correct answer to the question $(x_{1}, x_{2}, x_{3})$ can be simply computed as a sum of some of the numbers $t_{0}, \ldots, t_{7}$.

If $x_{1} = 0$, the players actually perform the quantum winning strategy for the Mermin-GHZ game with a systematic error only in the Hadamard transformation. 

Otherwise, the players win if and only if they measure one of the states $$|001\rangle, |010\rangle, |100\rangle, |111\rangle.$$ A straightforward but slightly tedious computation yields $$t_{1} = t_{2} = t_{4} = t_{7} = \frac{1}{4}(\frac{1}{2} + \sin(\frac{\pi}{4} + \frac{\epsilon_{2}}{2})\cos(\frac{\pi}{4} + \frac{\epsilon_{2}}{2})\cos\epsilon_{1}).$$ 

It follows that
\begin{eqnarray*}
               P_{\epsilon}(x_{1}, x_{2}, x_{3}) = 
                  \left\{ \begin{array}{ll}
                      \frac{1}{2} + \sin(\frac{\pi}{4} + \frac{\epsilon_{2}}{2})\cos(\frac{\pi}{4} + \frac{\epsilon_{2}}{2}) \ & \ \text{if $x_{1} = 0$}\\
                      \frac{1}{2} + \sin(\frac{\pi}{4} + \frac{\epsilon_{2}}{2})\cos(\frac{\pi}{4} + \frac{\epsilon_{2}}{2})\cos\epsilon_{1} \ & \text{otherwise}\\
                  \end{array} \right.
\end{eqnarray*}
and $$P_{\epsilon} = \frac{1}{2}(1 + \sin(\frac{\pi}{4} + \frac{\epsilon_{2}}{2})\cos(\frac{\pi}{4} + \frac{\epsilon_{2}}{2})(1 + \cos\epsilon_{1})).$$
\end{proof}
\begin{figure}[h]
    \centering
    \includegraphics[width = 7cm]{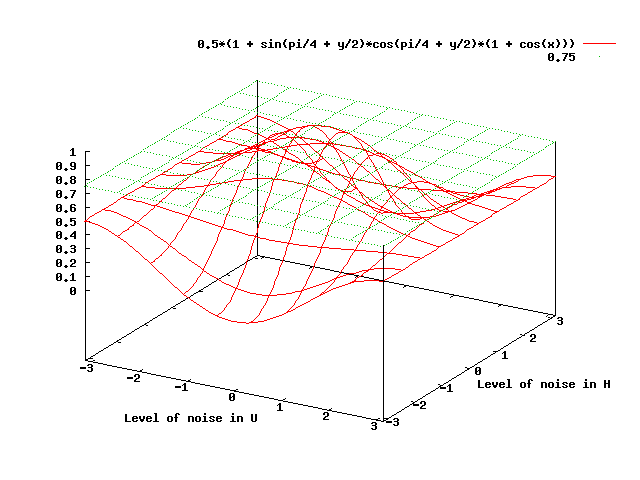}
    \includegraphics[width = 7cm]{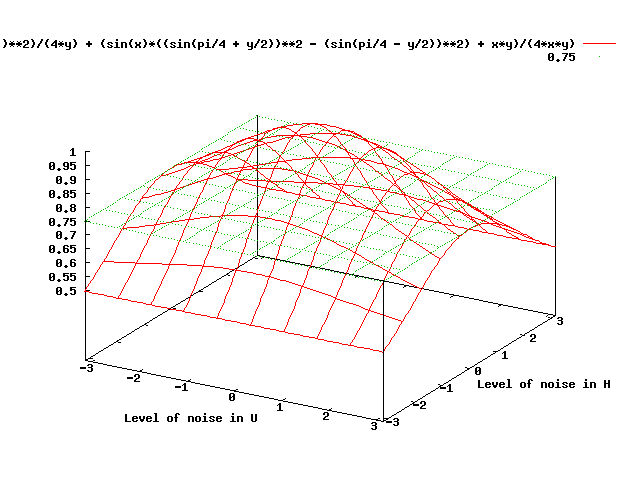}
\caption{Dependence of the success of quantum players on the level of unitary noise in the presence of a systematic (the graph on the left) and a random error made by one player in both unitary transformation.}     
\end{figure}
\subsection{Random errors}
\begin{theorem}
If Alice performs both transformations $U$, $H$ imperfectly with random errors with bounds $\delta_{1}$, $\delta_{2}$, respectively, then the quantum winning strategy for the Mermin-GHZ game gives a correct answer with probability
\begin{eqnarray*}
               P_{\delta}(x_{1}, x_{2}, x_{3}) = 
                  \left\{ \begin{array}{ll}
                      \frac{1}{2\delta}(\delta + \sin^{2}(\frac{\pi}{4} + \frac{\delta}{2}) - \sin^{2}(\frac{\pi}{4} - \frac{\delta}{2})) \ & \ \text{if $x_{1} = 0$}\\
                      \frac{sin\delta_{1}(\sin^{2}(\frac{\pi}{4} + \frac{\delta_{2}}{2}) - \sin^{2}(\frac{\pi}{4} - \frac{\delta_{2}}{2})) + \delta_{1}\delta_{2}}{2\delta_{1}\delta_{2}}
                       \ & \text{otherwise}\\
                  \end{array} \right.
\end{eqnarray*}
for any inputs $x_{1}$, $x_{2}$ and $x_{3}$ satisfying the promise. It gives a correct answer with probability $$P_{\delta} = \frac{1}{4}(1 + \frac{\sin^{2}(\frac{\pi}{4} + \frac{\delta_{2}}{2}) - \sin^{2}(\frac{\pi}{4} - \frac{\delta_{2}}{2})}{\delta_{2}} + \frac{sin\delta_{1}(\sin^{2}(\frac{\pi}{4} + \frac{\delta_{2}}{2}) - \sin^{2}(\frac{\pi}{4} - \frac{\delta_{2}}{2})) + \delta_{1}\delta_{2}}{\delta_{1}\delta_{2}}$$ for a question chosen uniformly and randomly from P.
\end{theorem}
\begin{proof}
Let Alice, Bob and Charles be given bits $x_{1}, x_{2}, x_{3}$, respectively, such that $\sum_{i = 1}^{3}x_{i}$ is divisible by $2$. Suppose that while Bob and Charles perform all the steps required by the quantum winning strategy for the Mermin-GHZ game correctly, Alice deviates in the first and the second step from the winning strategy by performing the unitary transformations 
\begin{eqnarray*}
\left. \begin{matrix}
U_{\epsilon_{1}}(|0\rangle) = |0\rangle, \\
U_{\epsilon_{1}}(|1\rangle) = e^{(\frac{\pi}{2} + \epsilon_{1}) \imath x_{i}}|1\rangle \\
\end{matrix} \right.,~~
H_{\epsilon_{2}} = \left( \begin{matrix}
    \cos(\frac{\pi}{4} + \frac{\epsilon_{2}}{2}) & \sin(\frac{\pi}{4} + \frac{\epsilon_{2}}{2}) \\
    \sin(\frac{\pi}{4} + \frac{\epsilon_{2}}{2}) & -\cos(\frac{\pi}{4} + \frac{\epsilon_{2}}{2}) \\
  \end{matrix} \right), 
\end{eqnarray*}  
respectively, where $\epsilon_{1}$, $\epsilon_{2}$ are chosen uniformly and randomly from the intervals $[-\delta_{1}, \delta_{1}]$, $[-\delta_{2}, \delta_{2}]$, respectively. The probability $P_{\delta}(x_{1}, x_{2}, x_{3})$ of the event that they obtain outcomes corresponding to a correct answer to the question $(x_{1}, x_{2}, x_{3})$ can be computed as $$P_{\delta}(x_{1}, x_{2}, x_{3}) = \frac{1}{4\delta_{1}\delta_{2}}\int\limits_{-\delta_{2}}^{\delta_{2}}\int\limits_{-\delta_{1}}^{\delta_{1}}P_{\epsilon}(x_{1}, x_{2}, x_{3})\, \mathrm{d} \epsilon_{1} \, \mathrm{d} \epsilon_{2}$$ where $P_{\epsilon}(x_{1}, x_{2}, x_{3})$ is the probability of a correct answer in the presence of a systematic error in both transformations. 

It follows that
\begin{eqnarray*}
               P_{\delta}(x_{1}, x_{2}, x_{3}) = 
                  \left\{ \begin{array}{ll}
                      \frac{1}{2\delta}(\delta + \sin^{2}(\frac{\pi}{4} + \frac{\delta}{2}) - \sin^{2}(\frac{\pi}{4} - \frac{\delta}{2})) \ & \ \text{if $x_{1} = 0$}\\
                      \frac{sin\delta_{1}(\sin^{2}(\frac{\pi}{4} + \frac{\delta_{2}}{2}) - \sin^{2}(\frac{\pi}{4} - \frac{\delta_{2}}{2})) + \delta_{1}\delta_{2}}{2\delta_{1}\delta_{2}}
                       \ & \text{otherwise.}\\
                  \end{array} \right.
\end{eqnarray*}
and $$P_{\delta} = \frac{1}{4}(1 + \frac{\sin^{2}(\frac{\pi}{4} + \frac{\delta_{2}}{2}) - \sin^{2}(\frac{\pi}{4} - \frac{\delta_{2}}{2})}{\delta_{2}} + \frac{\sin\delta_{1}(\sin^{2}(\frac{\pi}{4} + \frac{\delta_{2}}{2}) - \sin^{2}(\frac{\pi}{4} - \frac{\delta_{2}}{2})) + \delta_{1}\delta_{2}}{\delta_{1}\delta_{2}}.$$
\end{proof}
\section{Conclusions and Open Problems}
We have studied the impact of erroneously performed unitary transformations on the quantum winning strategy for the Mermin-GHZ game. It has turned out that these errors are able to decrease the success probability of quantum players so that they have no advantage over classical players, provided the unitary noise is sufficiently strong. We have also investigated how strong the noise can be so that quantum players would still be better than classical ones.

In this paper we have considered a fixed quantum strategy only. It would be certainly reasonable to investigate whether there is a quantum strategy which is better than the winning strategy in the presence of unitary noise and to find a bound on the success of quantum strategies. A natural direction of the future research is also to examine the impact of unitary noise on generalizations of the Mermin-GHZ game called Mermin's parity game \cite{eqeiasomds} and the extended parity game \cite{caqnl}. Another open question of interest is how unitary noise influences other pseudo-telepathy games.
\bibliographystyle{eptcs}

\end{document}